\documentclass[sn-mathphys]{sn-jnl}

\usepackage[ruled, vlined, nofillcomment, linesnumbered]{algorithm2e}
\usepackage{url}
\usepackage{amsmath}
\usepackage{amssymb}
\usepackage{graphicx}
\usepackage{subcaption}
\usepackage[english]{babel}
\usepackage{booktabs}
\usepackage{nicefrac}
\usepackage{verbatim}
\usepackage{float}
\usepackage{footmisc}
\usepackage{xspace}
\usepackage{tikz}
\usepackage{pgfplots}
\usepackage{pgfplotstable}

\usepackage[numbers,sort&compress]{natbib}

\usepackage{flushend}

\usepackage{etoolbox}

\makeatletter
\patchcmd{\@algocf@start}% <cmd>
  {-1.5em}% <search>
  {0pt}% <replace>
  {}{}% <success><failure>
\makeatother

\SetKw{Output}{output}

\newcommand{\set}[1]{\left\{#1\right\}}
\newcommand{\pr}[1]{\left(#1\right)}
\newcommand{\fpr}[1]{\mathopen{}\left(#1\right)}
\newcommand{\spr}[1]{\left[#1\right]}

\newcommand{\abs}[1]{{\left|#1\right|}}
\newcommand{\floor}[1]{{\left\lfloor#1\right\rfloor}}

\newcommand{\enpr}[2]{\pr{#1 ,\ldots , #2}}

\newcommand{\np}{\textbf{NP}}

\newcommand{\naturals}{\mathbb{N}}
\newcommand{\funcdef}[3]{{#1}:{#2} \to {#3}}
\newcommand{\define}{\leftarrow}

\newcommand{\fm}[1]{{\mathcal{#1}}}

\DeclareRobustCommand{\dispfunc}[2]{%
	\ensuremath{%
		\ifthenelse{\equal{#2}{}}%
			{\mathit{#1}}%
			{\mathit{#1}\fpr{#2}}}}

\newcommand{\bigO}[1]{\dispfunc{\mathcal{O}}{#1}}

\newcommand{\geom}[1]{\dispfunc{geo}{#1}}

\newcommand{\est}[1]{\dispfunc{\Delta}{#1}}

\newcommand{\dens}[1]{\dispfunc{dns}{#1}}

\newcommand{\dtname}[1]{\textsl{#1}}

\newcommand{\prbmatch}{\textsc{3Dmatch}\xspace}
\newcommand{\prbdense}{\textsc{Dense}\xspace}

\newcommand{\algcompute}{\textsc{Compute}\xspace}
\newcommand{\algupdate}{\textsc{Update}\xspace}
\newcommand{\algcore}{\textsc{Core}\xspace}
\newcommand{\algestimate}{\textsc{Estimate}\xspace}
\newcommand{\algexact}{\textsc{ExactCore}\xspace}
\newcommand{\algdense}{\textsc{EstDense}\xspace}

\newcommand{\alglb}{\textsc{lb}\xspace}
\newcommand{\alglub}{\textsc{lub}\xspace}

\newcommand{\algmergesort}{\textsc{MergeSort}\xspace}
\newcommand{\algunion}{\textsc{Union}\xspace}
\newcommand{\algdelete}{\textsc{Delete}\xspace}

\newtheorem{lemma}{Lemma}[section]
\newtheorem{proposition}{Proposition}[section]
\newtheorem{corollary}{Corollary}[section]
\newtheorem{definition}{Definition}[section]
\newtheorem{problem}{Problem}[section]

\SetKw{Break}{break}

\pgfdeclarelayer{background}
\pgfdeclarelayer{foreground}
\pgfsetlayers{background,main,foreground}

\definecolor{yafaxiscolor}{rgb}{0.3, 0.3, 0.3}

\definecolor{yafcolor1}{rgb}{0.4, 0.165, 0.553}
\definecolor{yafcolor2}{rgb}{0.949, 0.482, 0.216}
\definecolor{yafcolor3}{rgb}{0.47, 0.549, 0.306}
\definecolor{yafcolor4}{rgb}{0.925, 0.165, 0.224}
\definecolor{yafcolor5}{rgb}{0.141, 0.345, 0.643}
\definecolor{yafcolor6}{rgb}{0.965, 0.933, 0.267}
\definecolor{yafcolor7}{rgb}{0.627, 0.118, 0.165}
\definecolor{yafcolor8}{rgb}{0.878, 0.475, 0.686}

\newlength{\yafaxispad}
\newlength{\yaftlpad}
\newlength{\yaflabelpad}
\newlength{\yafaxiswidth}
\newlength{\yafticklen}
\makeatletter
\def\pgfplots@drawtickgridlines@INSTALLCLIP@onorientedsurf#1{}
\makeatother

\newcommand{\yafdrawaxis}[4]{
	\pgfplotstransformcoordinatex{#1}\let\xmincoord=\pgfmathresult 
	\pgfplotstransformcoordinatex{#2}\let\xmaxcoord=\pgfmathresult 
	\pgfplotstransformcoordinatey{#3}\let\ymincoord=\pgfmathresult 
	\pgfplotstransformcoordinatey{#4}\let\ymaxcoord=\pgfmathresult 
	\pgfsetlinewidth{\yafaxiswidth} 
	\pgfsetcolor{yafaxiscolor}
	\pgfpathmoveto{\pgfpointadd{\pgfpointadd{\pgfplotspointrelaxisxy{0}{0}}{\pgfqpointxy{\xmincoord}{0}}}{\pgfqpoint{-0.5\yafaxiswidth}{\yafaxispad}}}
	\pgfpathlineto{\pgfpointadd{\pgfpointadd{\pgfplotspointrelaxisxy{0}{0}}{\pgfqpointxy{\xmaxcoord}{0}}}{\pgfqpoint{0.5\yafaxiswidth}{\yafaxispad}}}
	\pgfpathmoveto{\pgfpointadd{\pgfpointadd{\pgfplotspointrelaxisxy{0}{0}}{\pgfqpointxy{0}{\ymincoord}}}{\pgfqpoint{\yafaxispad}{-0.5\yafaxiswidth}}}
	\pgfpathlineto{\pgfpointadd{\pgfpointadd{\pgfplotspointrelaxisxy{0}{0}}{\pgfqpointxy{0}{\ymaxcoord}}}{\pgfqpoint{\yafaxispad}{0.5\yafaxiswidth}}}
	\pgfusepath{stroke}
}

\pgfplotscreateplotcyclelist{yaf}{% 
{yafcolor5,mark options={scale=0.75},mark=o}, 
{yafcolor2,mark options={scale=0.75},mark=square},
{yafcolor3,mark options={scale=0.75},mark=triangle},
{yafcolor4,mark options={scale=0.75},mark=o},
{yafcolor1,mark options={scale=0.75},mark=o},
{yafcolor6,mark options={scale=0.75},mark=o},
{yafcolor7,mark options={scale=0.75},mark=o},
{yafcolor8,mark options={scale=0.75},mark=o}} 

\pgfkeys{/pgf/number format/.cd,1000 sep={\,}}

\pgfplotsset{axis y line=left, axis x line=bottom,
	tick align=outside,
	tickwidth=\yafticklen,
	clip = false,
    x axis line style= {-, line width = 0pt, color=black!0},
    y axis line style= {-, line width = 0pt, color=black!0},
    x tick style= {line width = \yafaxiswidth, color=yafaxiscolor, yshift = \yafaxispad},
    y tick style= {line width = \yafaxiswidth, color=yafaxiscolor, xshift = \yafaxispad},
    x tick label style = {font=\small, yshift = \yaftlpad, inner xsep = 0pt},
    y tick label style = {font=\small, xshift = \yaftlpad},
    every axis y label/.style = {at = {(ticklabel cs:0.5)}, rotate=90, anchor=center, font=\small, yshift = -\yaflabelpad, inner sep = 0pt},
    every axis x label/.style = {at = {(ticklabel cs:0.5)}, anchor=center, font=\small, yshift = \yaflabelpad},
    x tick label style = {font=\small, yshift = 1pt},
    grid = major,
    major grid style  = {dash pattern = on 1pt off 3 pt},
	every axis plot post/.append style= {line width=\yafaxiswidth} ,
	legend cell align = left,
	legend style = {inner sep = 1pt, cells = {font=\scriptsize}},
	legend image code/.code={%
		\draw[mark repeat=2,mark phase=2,#1] 
		plot coordinates { (0cm,0cm) (0.15cm,0cm) (0.3cm,0cm) };% 
	} 
}

\newcommand{\pgfprintduration}[1]{%
	\ifthenelse{\equal{#1}{}}{---}{%
	\pgfmathsetmacro{\minutes}{floor(#1 / 60)}%
	\pgfmathsetmacro{\seconds}{#1 - 60*\minutes}%
	\pgfmathifthenelse{\minutes > 0}{"\pgfmathprintnumber{\minutes}m \pgfmathprintnumber[fixed,precision=0]{\seconds}s"}{"\pgfmathprintnumber{\seconds}s"}\pgfmathresult}}

\begin{document}

\author[1]{\fnm{Nikolaj} \sur{Tatti}}\email{nikolaj.tatti@helsinki.fi}

\affil[1]{\orgdiv{HIIT}, \orgname{University of Helsinki}, \orgaddress{\country{Finland}}}

\title{Fast computation of distance-generalized cores using sampling}

\abstract{
Core decomposition is a classic technique for discovering densely connected
regions in a graph with large range of applications. Formally, a $k$-core
is a maximal subgraph where each vertex has at least $k$ neighbors.
A natural extension of a $k$-core is a 
$(k, h)$-core, where each node must have at least $k$ nodes
that can be reached with a path of length $h$.
The downside in using $(k, h)$-core decomposition is the significant
increase in the computational complexity: whereas the standard core decomposition can
be done in $\bigO{m}$ time, the generalization can require $\bigO{n^2m}$ time, where
$n$ and $m$ are the number of nodes and edges in the given graph.

In this paper we propose a randomized algorithm
that produces an $\epsilon$-approximation of $(k, h)$ core decomposition with a probability of $1 - \delta$ in
$\bigO{\epsilon^{-2} hm (\log^2 n - \log \delta)}$ time.
The approximation is based on sampling the neighborhoods of nodes, and we use Chernoff bound
to prove the approximation guarantee.
We also study distance-generalized dense subgraphs, show that the problem is \np-hard,
provide an algorithm for discovering such graphs with approximate core decompositions,
and provide theoretical guarantees for the quality of the discovered subgraphs.
We demonstrate empirically that approximating the decomposition
complements the exact computation: computing the approximation is significantly faster than computing the exact solution
for the networks where computing the exact solution is slow.}


\maketitle

\section{Introduction}
% cores

Core decomposition is a classic technique for discovering densely connected
regions in a graph. The appeal of core decomposition is a simple and intuitive definition,
and the fact that the core decomposition can be computed in linear time.
Core decomposition has a large range of applications such as
graph visualization~\citep{alvarez2006large},
graph modeling~\citep{bollobas1984evolution}, social network analysis~\citep{seidman1983network},
internet topology modeling~\citep{Carmi03072007},
influence analysis~\citep{kitsak10influence,Ugander17042012}, bioinformatics~\citep{bader2003automated,hagmann08brain},
and team formation~\citep{Bonchi:2014kh}.

More formally,
a $k$-core is a maximal subgraph such that every vertex has at least $k$ degree.
We can show that $k$-core form a nested structure: $(k + 1)$-core is a subset of $k$-core,
and that the core decomposition can be discovered in linear time~\citep{batagelj2011fast}.
Core decomposition has been extended to directed~\citep{giatsidis2013d}, multi-layer~\citep{galimberti2017core},
temporal~\citep{galimberti2018mining}, and
weighted~\citep{serrano2009extracting}
networks.

% distance cores
A natural extension of core decomposition, proposed by~\citet{bonchi2019distance}, is a distance-generalized core
decomposition or $(k, h)$-core decomposition, where the degree is replaced by the number of nodes that
can be reached with a path of length $h$. Here, $h$ is a user parameter
and $h = 1$ reduces to a standard core decomposition.
Using distance-generalized core decomposition may produce a more refined
decomposition~\citep{bonchi2019distance}. Moreover, it can be used
when discovering $h$-clubs, distance-generalized dense subgraphs,
and distance-generalized chromatic numbers~\citep{bonchi2019distance}.

Studying such structures may be useful in graphs where paths of length $h$
reveal interesting information.  For example, assume a authorship network,
where an edge between a paper an a researcher indicate that the researcher was
an author of the paper. Then paths of length $2$ contain co-authorship
information. 

The major downside in using the distance-generalized core decomposition is the significant 
increase in the computational complexity: whereas the standard core decomposition can
be done in $\bigO{m}$ time, the generalization can require $\bigO{n^2m}$ time, where
$n$ and $m$ are the number of nodes and edges in the given graph.

To combat this problem we propose a randomized algorithm
that produces an $\epsilon$-approximation of $(k, h)$ core decomposition with a probability of $1 - \delta$ in
\[
	\bigO{\frac{hm\log n / \delta}{\epsilon^{2}}\log \frac{n\epsilon^{2}}{\log n / \delta}}
	\subseteq \bigO{\epsilon^{-2} hm (\log^2 n - \log \delta)}
\]
time.

The intuition behind our approach is as follows.
In order to compute the distance-generalized core decomposition we need
to discover and maintain $h$-neighborhoods for each node.
We can discover the
$h$-neighborhood of a node $v$ by taking the union of the $(h -
1)$-neighborhood of the adjacent nodes, which leads to a simple dynamic
program. The computational bottleneck comes from the fact that these
neighborhoods may become too large.  So, instead of computing the complete
neighborhood, we have a carefully selected budget $M$. The moment the
neighborhood becomes too large, we delete (roughly) half of the nodes, and to
compensate for the sampling we multiply our size estimate by 2. This procedure
is repeated as often as needed.  Since we are able to keep the neighbor samples
small, we are able to compute the decomposition faster. 

We use Chernoff bounds to determine an appropriate value for $M$, and provide
algorithms for maintaining the $h$-neighborhoods. The maintenance require special
attention since if the $h$-neighborhood becomes too small we need to bring back
the deleted nodes. 

Finally, we study distance-generalized subgraphs, a notion proposed
by~\citet{bonchi2019distance} that extends a notion of dense subgraphs. Here
the density is the ratio of $h$-connected node pairs and nodes. We show that
the problem is \np-hard and propose an algorithm based on approximate core maps,
extending the results by~\citet{bonchi2019distance}.

The rest of the paper is organized as follows.  In Section~\ref{sec:prel} we
introduce preliminary notation and formalize the problem.  In
Section~\ref{sec:algorithm} we present a naive version of the algorithm that
yields approximate results but is too slow.  We prove the approximation guarantee
in Section~\ref{sec:sampling}, and speed-up the algorithm in
Section~\ref{sec:speedup}. 
In Section~\ref{sec:dense} we study distance-generalized dense subgraphs.
We discuss the related work in
Section~\ref{sec:related}. Finally, we compare our method empirically against
the baselines in Section~\ref{sec:exp} and conclude the paper with discussion
in Section~\ref{sec:conclusions}.

This work extends the conference paper~\citep{tatti2021khcore}. 

\section{Preliminaries and problem definition}\label{sec:prel}
In this section we establish preliminary notation and define our problem.

Assume an undirected graph $G = (V, E)$ with $n$ nodes and $m$ edges. We will write
$A(v)$ to be the set of nodes adjacent to $v$. Given an integer $h$, we define an
\emph{$h$-path} to be a sequence of \emph{at most} $h + 1$ adjacent nodes. An
\emph{$h$-neighborhood} $N(v; h, X)$ is then the set of nodes that are reachable with an
$h$-path in a set of nodes $X$. If $X = V$ or otherwise clear from context, we will drop it from the notation.
Note that $N(v; 1) = A(v) \cup \set{v}$.

We will write $\deg_h(v; X) = \abs{N(v; h, X)} - 1$, where $X$ is 
a set of nodes and $v \in X$.
We will often drop $X$ from the notation if it is clear from the context.

A \emph{$k$-core} is the maximal subgraph of $G$ for which all nodes have at least a
degree of $k$.
Discovering the cores can be done in $\bigO{m}$
time by iteratively deleting the vertex with the smallest degree~\citep{seidman1983network}.

\citet{bonchi2019distance} proposed to extend the notion of $k$-cores to $(k, h)$-cores.
Here, given an integer $h$, a \emph{$(k, h)$-core} is the maximal graph $H$ of $G$
such that $\abs{N(v; h)} - 1 \geq k$ for each $v \in V(H)$, that is, we can reach at least $k$ nodes from $v$
with a path of at most $h$ nodes. The \emph{core number} $c(v)$ of a vertex $v$ is the largest $k$ such that $v$
is contained in $(k, h)$-core $H$. We will call $H$ as the \emph{core graph} of $v$ and
we will refer to $c$ as the \emph{core map}.

Note that discovering $(k, 1)$-cores is equal to discovering standard
$k$-cores.  We follow the same strategy when computing $(k, h)$-cores as with
standard cores: we iteratively find and delete the vertex with the smallest
degree~\citep{bonchi2019distance}. We will refer to the exact algorithm as \algexact.
While \algexact is guaranteed to produce the correct result the computational
complexity deteriorates to $\bigO{n^2m}$. The main reason here is that the neighborhoods
$N(v; h)$ can be significantly larger than just adjacent nodes $A(v)$.

In this paper we consider approximating cores.
\begin{definition}[approximative $(k, h)$-core]
Given a graph $G$ an integer $h$ and approximation guarantee $\epsilon$, an
$\epsilon$-approximative core map $\funcdef{c'}{V}{\naturals}$ maps a node
to an integer such that $\abs{c'(v) - c(v)} \leq \epsilon c(v)$ for each $v \in V$.
\end{definition}

We will introduce an algorithm that computes an $\epsilon$-approximative core map
with a probability of $1 - \delta$ in quasilinear time.

\section{Naive, slow algorithm}\label{sec:algorithm}

In this section we introduce a basic idea of our approach. This version of the algorithm 
will be still too slow but will approximate the cores accurately. We will prove the accuracy
in the next section, and then refine the subroutines to obtain the needed computational complexity.

The bottleneck for computing the cores is maintaining the $h$-neighborhood $N(v; h)$ for
each node $v$ as we delete the nodes. Instead of maintaining the complete $h$-neighborhood
we will keep only certain nodes if the neighborhood becomes too large. We then compensate
the sampling when estimating the size of the $h$-neighborhood.

Assume that we are given a graph $G$, an integer $h$, approximation guarantee $\epsilon$,
and a probability threshold $\delta$. Let us define numbers $C = \log(2n  / \delta)$ and
\begin{equation}
\label{eq:buffer}
    M = 1 + \frac{4(2 + \epsilon)}{\epsilon^2}(C + \log 8) \quad.
\end{equation}
The quantity $M$ will act as an upper bound for the sampled $h$-neighborhood, while $C$
will be useful when analyzing the properties of the algorithm.
We will see later that these specific values will yield the approximation guarantees.

We start the algorithm by sampling the rank of a node from a geometric distribution
$r[v] = \geom{1/2}$. Note that ties are allowed.
During the algorithm we maintain two key
variables $B[v, i]$ and $k[v, i]$ for each $v \in V$ and each index $i = 1, \ldots, h$.
Here, 
\[
	B[v, i] = \set{u \in N(v; i) \mid r[u] \geq k[v, i]}
\]
is a subset of $i$-neighborhood $N(v; i)$ consisting of nodes whose rank $r[u] \geq k[v, i]$. 
The threshold $k[v, i]$ is set to be as small as possible such that $\abs{B[v, i]} \leq M$.

We can estimate $c(v)$ from $B[v, h]$ and $k[v, h]$ as follows:
Consider the quantity $d = \abs{B[v, h] \setminus \set{v}}2^{k[v, h]}$.
Note that for an integer $k$ the probability of a vertex $v$ having a rank $r[v] \geq
k$ is $2^{-k}$. This hints that $d$ is a good estimate for $c(v)$.
We show in the next section that this is indeed the case but $d$ is lacking an
important property that we need in order to prove the correctness of the
algorithm. 
Namely, $d$ can increase while we are deleting nodes. To fix this pathological case
we estimate $c(v)$ with $\max(d, M2^{k[v, h] - 1})$ if $k[v, h] > 0$, and with $d$ if $k[v, h] = 0$.
The pseudo-code for the estimate is given in Algorithm~\ref{alg:estimate}.

\begin{algorithm}
\caption{$\algestimate(v)$, estimates $\abs{N(v, h)} - 1$ using $B[v, h]$ and $k[v, h]$.}
\label{alg:estimate}
$k \define k[v, h]$\;
$d \define  \abs{B[v, h] \setminus \set{v}}2^k$\;
\lIf {$k > 0$} {$d \define \max(d, M2^{k - 1})$}
\Return $d$\;
\end{algorithm}

To compute $B[v, i]$ we have the following observation.
\begin{proposition}
\label{prop:dyn}
For any $v \in V$ and any $i = 1, \ldots, h$, we have
\[
\begin{split}
	B[v, i] & = \set{u \in T \mid r[u] \geq k[v, i]}, \quad\text{where}\quad \\
	T & = \set{v} \cup \set{u \in B[w, i - 1] \mid w \in A(v)}\quad.
\end{split}
\]
Moreover, $k[v, i] \geq k[w, i - 1]$ for any $w \in A(v)$.
\end{proposition}

\begin{proof}
Let $w \in A(v)$. Since $N(w, i - 1) \subseteq N(v, i)$, we have $k[v, i] \geq k[w, i - 1]$.
Consequently, $B[v, i] \subseteq T \subseteq N(v, i)$, and by definition of $B[v, i]$, the claim follows.
\end{proof}

The proposition leads to \algcompute, an algorithm for computing $B[v, i]$ given in
Algorithm~\ref{alg:naivecompute}. Here, we form a set $T$, a union of sets $B[w, i - 1]$, where $w \in A(v)$.
After $T$ is formed we search for the threshold $k[v, i] \geq \max_{w \in A(v)} k[w, i - 1]$ that yields at most $M$ nodes in $T$,
and store the resulting set in $B[v, i]$.

\begin{algorithm}
\caption{Naive version of $\algcompute(v, i)$. Recomputes $B[v, i]$ and $k[v, i]$ from scratch.}
\label{alg:naivecompute}

$T \define (t_j)_{j = 1} \define \set{v} \cup \bigcup_{w \in A(v)} B[w, i - 1]$ sorted by $r[\cdot]$\;
$k \define \max \set{k[w, i - 1] \mid w \in A(v)}$\;

\If {$\abs{T} > M$} {
	$k \define \max(k, r[t_{\floor{M} + 1}] + 1)$\;
}

$B[v, i] \define \set{u \in T \mid r[u] \geq k}$\;
$k[v, i] \define k$\;

\end{algorithm}

As the node, say $u$, is deleted we need to update the affected nodes.
We do this update in Algorithm~\ref{alg:update} by
recomputing the neighbors $v \in A(u)$, and see if $B[v, i]$ and $k[v, i]$ have changed;
if they have, then we recompute $B[w, i + 1]$ for all $w \in A(v)$, and so on. 

\begin{algorithm}
\caption{Naive version of $\algupdate(u)$. Deletes $u$ and updates the affected $B[v, i]$ and $k[v, i]$.}
\label{alg:update}

delete $u$ from $G$\;

$U \define \emptyset$\;
\ForEach{$i = 1,\ldots, h$} {
	add neighbors of $u$ to $U$\;
	$W \define \emptyset$\;
	\ForEach {$v \in U$} {
		$\algcompute(v, i)$\;
		\If {$B[v, i]$ or $k[v, i]$ has changed} {
			add neighbors of $v$ in $W$\;
		}
        \lIf {$i = h$} {
            $d[v] \define \algestimate(v)$
        }
	}
	$U \define W$\;
}
\end{algorithm}

The main algorithm \algcore, given in Algorithm~\ref{alg:core}, initializes
$B[v, i]$ using \algcompute, deletes iteratively the nodes with smallest estimate $d[v]$ while updating the sets $B[v, i]$
with \algupdate.

\begin{algorithm}
\caption{$\algcore(G, \epsilon, C)$ approximative core decomposition. Setting
$C = \log(2n / \delta)$ yields an $\epsilon$-approximation with $1 - \delta$ probability.}
\label{alg:core}

\ForEach {$v \in V$} {
	$r[v] \define \geom{1/2}$\;
	$B[v, 0] \define \set{v}$\;
	$k[v, 0] \define 0$\;
}

$M \define 1 + \frac{4(2 + \epsilon)}{\epsilon^2}(C + \log 8)$\; 

\ForEach {$i = 1, \ldots, h$} {
	\lForEach {$v \in V$} {
		$\algcompute(v, i)$
	}
}
$c \define 0$\;
\While {graph is not empty} {
	$u \define \arg \min_v d[v]$ (use $k[v, h]$ as a tie breaker)\;
	$c \define \max(c, d[u])$\;
	output $u$ with $c$ as the core number\;
	$\algupdate(u)$\;
}

\end{algorithm}

\section{Approximation guarantee}
\label{sec:sampling}

In this section we will prove the approximation guarantee of our algorithm.
The key step is to show that \algestimate produces an accurate estimate.
For notational convenience, we need the following definition.

\begin{definition}
\label{def:est}
Assume $d$ integers $X = \enpr{x_1}{x_d}$ and an integer $M$. Define
\[
	S_i = \abs{\set{j \in \spr{d} \mid x_j \geq i}} \text{ and }
	T_i = \abs{\set{j \in \spr{d} \mid x_j \geq i, j \geq 2}}
\]
to be the number of integers larger than or equal to $i$.
Let $k \geq 0$ be the smallest integer for which $S_k \leq M$.
Define
\[
	\est{X; M} =
	\begin{cases}
	\max\pr{T_k2^k, M2^{k - 1}}, & \text{ if } k > 0, \\
	T_k2^k, & \text{ if } k = 0\quad.\\
	\end{cases}
\]
\end{definition}

Note that
if $R = \pr{r[w] \mid w \in N(v; h)}$ with $r[v]$ being the first element in $R$,
then $\est{R; M}$ is equal to the output of $\algestimate(v)$.

Our first step is to show that $\est{X; M}$ is monotonic.

\begin{proposition}
\label{prop:monotone}
Assume $M > 0$.
Let $x_1, \ldots, x_d$ be a set of integers.
Select $d' \leq d$.
Then
\[
	\est{x_1, \ldots, x_{d'}; M} \leq \est{x_1, \ldots, x_d; M}\quad.
\]
\end{proposition}

Note that this claim would not hold if we did not have the $M2^{k - 1}$ term in
the definition of $\est{X; M}$.

\begin{proof}
Let $k$, $S_i$, and $T_i$ be as defined for $\est{x_1, \ldots, x_d; M}$ in Definition~\ref{def:est}.
Also, let $k'$, $S_i'$, and $T_i'$ be as defined for $\est{x_1, \ldots, x_{d'}; M}$ in Definition~\ref{def:est}.

Since $S_i' \leq S_i$, we have $k' \leq k$.
If $k' = k$, the claim follows immediately since also $T_i' \leq T_i$. If $k' < k$, then
\[
	\est{x_1, \ldots, x_d; M} \geq M2^{k - 1} \geq M2^{k'} \geq T_{k'}' 2^{k'}
\]
and
\[
	\est{x_1, \ldots, x_d; M} \geq M2^{k - 1} \geq M2^{k' - 1},
\]
proving the claim.
\end{proof}

Next we formalize the accuracy of $\est{X; M}$.
We prove the claim in Appendix.

\begin{proposition}
\label{prop:approx}
Assume $0 < \epsilon \leq 1/2$. 
Let $\mathcal{R} = R_1, \ldots, R_d$ be independent random variables sampled from geometric distribution, $\geom{1/2}$.
Assume $C > 0$ and define $M$ as in Eq.~\ref{eq:buffer}.
Then
\begin{equation}
\label{eq:appr}
	\abs{\est{\mathcal{R}; M} - (d - 1)} \leq \epsilon (d - 1)
\end{equation}
with probability $1 - \exp\pr{-C}$. 
\end{proposition}

We are now ready to state the main claim.

\begin{proposition}
\label{prop:yield}
Assume graph $G$ with $n$ nodes, $\epsilon > 0$, and $C > 0$.
For each node $v \in V$, let $c(v)$ be the core number reported
by \algexact and let $c'(v)$ be the core number reported by
\algcore. Then with probability $1 - 2ne^{-C}$ 
\[
	\abs{c(v) - c'(v)} \leq \epsilon c(v),
\]
for every node in $V$.
Moreover, if $c(v) \leq M$, where $M$ is given in Eq.~\ref{eq:buffer},
then $c(v) = c'(v)$.
\end{proposition}

We will prove the main claim of the proposition with two lemmas.
In both proofs we will use the variable $\tau_v$ which we define
to be the value of $d[v]$ when $v$ is deleted by \algcore.

The first lemma establishes a lower bound.

\begin{lemma}
The lower bound 
	$c'(v) \geq (1 - \epsilon)c(v)$
holds
with probability $1 - ne^{-C}$.
\end{lemma}

\begin{proof}
For each node $v \in V$, let $R_v$ be a rank, an independent random variable sampled
from geometric distribution, $\geom{1/2}$.

Let $H_v$ be the core graph of $v$ as solved by \algexact.
Define $S_v = N(v, h) \cap H_v$ to be the $h$-neighborhood of $v$ in $H_v$.
Note that $c(v) \leq \abs{S_v} - 1$.
Let $\mathcal{R}_v$ be the list of ranks $\pr{R_w ; w \in S_v}$
such that $R_v$ is always the first element.

Proposition~\ref{prop:approx} combined with the union bound states that
\begin{equation}
\label{eq:upapp}
	\abs{\est{\mathcal{R}_v; M} - (\abs{S_v} - 1)} \leq \epsilon (\abs{S_v} - 1)\quad.
\end{equation}
holds with probability $1 - ne^{-C}$ for every node $v$.
Assume that these events hold.

To prove the claim, select a node $v$
and let $w$ be the first node in $H_v$ deleted by \algcore.
Let $F$ be the graph right before deleting $w$ by \algcore.
Then
\begin{align*}
	c'(v) & \geq c'(w) & \text{\hspace{-1cm}(\algcore picked $w$ before $v$ or $w = v$)}\\
	& \geq \tau_w  \\
	& \geq \est{\mathcal{R}_w; M} & \text{\hspace{-1cm}($H_w \subseteq H_v \subseteq F$ and Prop.~\ref{prop:monotone})} \\
	& \geq (1 - \epsilon) (\abs{S_w} - 1) & \text{(Eq.~\ref{eq:upapp})} \\
	& \geq (1 - \epsilon)c(w) & \text{($S_w = N(w, h) \cap H_w$)} \\
	& \geq (1 - \epsilon) c(v), & \text{($w \in H_v$)} \\
\end{align*}
proving the lemma.
\end{proof}

Next, we establish the upper bound.

\begin{lemma}
The upper bound 
	$c'(v) \leq (1 + \epsilon)c(v)$
holds
with probability $1 - ne^{-C}$.
\end{lemma}

\begin{proof}
For each node $v \in V$, let $R_v$ be an independent random variable sampled
from geometric distribution, $\geom{1/2}$.

Consider the exact algorithm \algexact for solving the $(k, h)$ core problem.
Let $H_v$ be the graph induced by the
existing nodes right before $v$ is removed by \algexact. 
Define $S_v = N(v, h) \cap H_v$ to be the $h$-neighborhood of $v$ in $H_v$.
Note that $c(v) \geq \abs{S_v} - 1$.
Let $\mathcal{R}_v$ be the list of ranks $\pr{R_w ; w \in S_v}$
such that $R_v$ is the first element.

Proposition~\ref{prop:approx} combined with the union bound states that
\begin{equation}
\label{eq:downapp}
	\abs{\est{\mathcal{R}_v; M} - (\abs{S_v} - 1)} \leq \epsilon (\abs{S_v} - 1)\quad.
\end{equation}
holds with probability $1 - ne^{-C}$ for every node $v$.
Assume that these events hold.

Select a node $v$. Let $W$ be the set containing $v$ and the nodes selected \emph{before} $v$ by \algcore.
Select $w \in W$.
Let $F$ be the graph right before deleting $w$ by \algcore. 
Let $u$ be the node in $F$ that is deleted first by \algexact.
Let $\beta$ be the value of $d[u]$ when $w$ is deleted by \algcore.
Then
\begin{align*}
	\tau_w & \leq \beta & \text{(\algcore picked $w$ over $u$ or $w = u$)}\\
	& \leq  \est{\mathcal{R}_u; M}   & \text{($F \subseteq H_u$ and Proposition~\ref{prop:monotone})}\\ 
	& \leq  (1 + \epsilon) (\abs{S_u} - 1) & \text{(Eq.~\ref{eq:downapp})}\\
	& \leq (1 + \epsilon) c(u) & \text{($S_u = N(u, h) \cap H_u$)} \\
	& \leq (1 + \epsilon) c(v) \quad. & \text{($v \in F \subseteq H_u$)}
\end{align*}

Since this bound holds for \emph{any} $w \in W$, we have
\[
	c'(v) = \max_{w \in W} \tau_w \leq (1 + \epsilon)c(v),
\]
proving the lemma.
\end{proof}

We are now ready to prove the proposition.

\begin{proof}[Proof of Proposition~\ref{prop:yield}]
The probability that one of the two above lemmas does not hold
is bounded by the union bound with $2ne^{-C}$, proving the main claim.

To prove the second claim note that when $d[v] \leq M$ then
$d[v]$ matches accurately the number of the remaining nodes that can be reached
by an $h$-path from a node $v$. On the other hand, if there is a node $w$ that reaches more than $M$
nodes, we are guaranteed that $d[w] \geq M$ and $k[w, h] > 0$, implying that \algcore will always prefer
deleting $v$ before $w$.
Consequently, at the beginning \algcore will select the nodes in the same order as \algexact and reports the same core number
as long as there are nodes with $d[v] \leq M$ or, equally, as long as $c(v) \leq M$.
\end{proof}

\section{Updating data structures faster}\label{sec:speedup}

Now that we have proven the accuracy of \algcore, our next step is to address
the computational complexity. The key problem is that \algcompute is called too often
and the implementation of \algupdate is too slow.

As \algcore progresses, the set $B[v, i]$ is modified in two ways. The first case is
when some nodes become too far away, and we need to delete these 
nodes from $B[v, i]$. The second case is when we have deleted enough nodes so that
we can lower $k[v, i]$ and introduce new nodes. Our naive version of \algupdate calls \algcompute for both cases.
We will modify the algorithms so that \algcompute is called only to handle the second case, and the first case is handled separately.
Note that these modifications do not change the output of the algorithm.

First, we change the information stored in $B[v, i]$. Instead of storing just a
node $u$, we will store $(u, z)$, where $z$ is the number of neighbors $w \in
A(v)$, such that $u$ is in $B[w, i - 1]$.
We will store $B[v, i]$ as a linked list \emph{sorted} by the rank.
In addition, each node $u \in B[w, i - 1]$ is augmented with an array $Q = (q_v \mid v \in A(w))$. 
An entry $q_v$ points to the location of $u$ in $B[v, i]$ if $u$ is present in $B[v, i]$.
Otherwise, $q_v$ is set to null.

We will need two helper functions to maintain $B[v, i]$. The first function
is a standard merge sort, $\algmergesort(X, Y)$, that merges two sorted lists
in $\bigO{\abs{X} + \abs{Y}}$ time, maintaining the counters and the pointers. 

The other function is $\algdelete(X, Y)$ that removes nodes in $Y$ from $X$,
which we will use to remove nodes from $B[v, i]$.
The deletion is done in by reducing the counters of the corresponding nodes in
$X$ by 1, and removing them when the counter reaches 0. It is vital that
we can process a single node $y \in Y$ in constant time. This will be possible
because we will be able to use the pointer array described above.

Let us now consider calling \algcompute. We would like to minimize the number of 
calls of \algcompute.
In order to do that, we need additional
bookkeeping. The first additional information is $m[v, i]$ which is the number of
neighboring nodes $w \in A(v)$ for which $k[w, i - 1] = k[v, i]$.
Proposition~\ref{prop:dyn} states 
that $k[v, i] \geq k[w, i - 1]$, for all $w \in A(v)$.
Thus, if $m[v, i] > 0$, then there is a node $u \in A(v)$ with $k[v, i] = k[u, i - 1]$
and so recomputing $B[v, i]$ will not change $k[v, i]$ and will not add new nodes in $B[v, i]$.

Unfortunately, maintaining just $m[v, i]$ is not enough. We may have $k[v, i] >
k[w, i - 1]$ for any $w \in A(v)$ \emph{immediately} after \algcompute. In such case,
we compute sets of nodes
\[
	X_w = \set{u \in B[w, i - 1] \mid r[u] = k[v, i] - 1},
\]
and combine them in $D[v, i]$, a union of $X_w$ along with the counter information 
similar to $B[v, i]$, that is,
\[
	D[v, i] = \set{(u, z) \mid z = \abs{\set{w \in A(v) \mid u \in X_w}} > 0}\quad.
\]
The key observation is that as long as $\abs{B[v, i]} + \abs{D[v, i]} > M$,
the level $k[v, i]$ does not need to be updated. 

There is one complication, namely, we need to compute $D[v, i]$ in $\bigO{M\deg v}$ time.
Note that, unlike $B[v, i]$, the set $D[v, i]$ can have more than $M$ elements.
Hence, using \algmergesort will not work. Moreover, a stock $k$-ary merge sort requires $\bigO{M\deg(v) \log \deg(v)}$ time.
The key observation to avoid the additional $\log$ factor 
is that $D[v, i]$ does not need to be sorted.
More specifically, we first compute an array
\[
	a[u] = \abs{\set{w \in A(v) \mid u \in X_w}},
\]
and then extract the non-zero entries to form $D[v, i]$.
We only need to compute the non-zero entries so we can compute these entries in
$\bigO{\sum \abs{X_w}} \subseteq \bigO{M\deg v}$ time. Moreover, since we do
not need to keep them in order we can extract the non-zero entries in the same time.
We will refer this procedure as \algunion, taking the sets $X_w$ as input and forming $D[v, i]$.

We need to maintain $D[v, i]$ efficiently. In order to do that we augment
each node $u \in B[w, i - 1]$ with an array $(q_v \mid v \in A(w))$, where
$q_v$ points to the location of $u$ in $D[v, i]$ if $u \in D[v, i]$.

The pseudo-code for the updated \algcompute is given in Algorithm~\ref{alg:compute2}.
Here we compute $B[v, i]$ and $k[v, i]$ first by using \algmergesort iteratively and trimming the resulting set if it has more than $M$ elements.
We proceed to compute $m[v, i]$ and $D[v, i]$. If $m[v, i] = 0$, we compute $D[v, i]$ with \algunion.
Note that if $m[v, i] > 0$, we leave $D[v, i]$ empty.
The above discussion leads immediately to the computational complexity of \algcompute.

\begin{proposition}
$\algcompute(v, i)$ runs in $\bigO{M \deg v}$ time.
\end{proposition}

The pseudo-code for \algupdate is given in Algorithm~\ref{alg:update2}.
Here, we maintain a stack $U$ of tuples $(v, Y)$, where $v$ is the node that requires an update,
and $Y$ are the nodes that have been deleted from $B[v, i]$ during the previous round. First, if $\abs{B[v, i]} + \abs{D[v, i]} \leq M$
and $m[v, i] = 0$, we run $\algcompute(v, i)$. Next, we proceed by reducing the counters of $Z$
in $B[w, i + 1]$ and $D[w, i + 1]$ for each $w \in A(v)$. We also update $m[w, i + 1]$.
Finally, we add $(w, Z)$ to the next stack, where $Z$ are the deleted nodes in $B[w, i + 1]$.

\begin{proposition}
\algupdate maintains $B[v, i]$ correctly.
\end{proposition}

\begin{proof}
As \algcore deletes nodes from the graph,
Proposition~\ref{prop:dyn} guarantees that
$B[v, i]$ can be modified only in two ways: either node $u$ is deleted
from $B[v, i]$ when $u$ is no longer present in any $B[w, i - 1]$ where $w \in A(v)$,
or $k[v, i]$ changes and new nodes are added.

The first case is handled properly as \algupdate uses \algdelete
whenever a node is deleted from $B[w, i - 1]$.

The second case follows since if $\abs{B[v, i]} + \abs{D[v, i]} > M$
or $m[v, i] > 0$, then we know that \algcompute will not change $k[v, i]$ and will not introduce new nodes in $B[v, i]$.
\end{proof}

\begin{algorithm}

\caption{Refined version of $\algcompute(v, i)$. Recomputes $B[v, i]$ and $k[v, i]$ from scratch.}
\label{alg:compute2}

$T \define {(v, 1)}$\;

$k \define \max\pr{k[w, i - 1] \mid w \in A(v)}$\;

\ForEach {$w \in A(v)$} {
	$T \define \algmergesort(T, \set{(u, c) \in B[w, i - 1] \mid r[u] \geq k})$\;
	\If {$\abs{T} > M$} {
		$k \define \max(k, r[t_{\floor{M} + 1}] + 1)$
		\tcp*{$(t_j) = T$}
		$T \define \set{u \in T \mid r[u] \geq k}$\;
	}
}

$B[v, i] \define T$\;
$k[v, i] \define k$\;

$m[v, i] \define \abs{\set{w \in A(v) \mid k[w, i - 1] = k[v, i]}}$\;
$D[v, i] \define \emptyset$\;
\If {$m[v, i] = 0$} {
	$X_w \define \set{u \in B[w, i - 1] \mid r[u] = k - 1}$ for $w \in A(v)$\;
	$D[v, i] \define \algunion(\set{X_w \mid w \in A(v)})$\;
}

\end{algorithm}

\begin{algorithm}

\caption{Refined version of $\algupdate(u)$. Deletes $u$ and updates the affected $B[v, i]$ and $k[v, i]$.}
\label{alg:update2}

$U \define \emptyset$\;

\ForEach{$i = 1, \ldots, h$} {
	add $(u, B[u, i - 1])$ to $U$\;
	$B[u, i - 1] \define \emptyset$\;
	$W \define \emptyset$\;
	\While {$U$ is not empty} {
		$(v, Y) \define $ pop entry from $U$\nllabel{alg:setx}\;

		$k \define k[v, i]$\;
		\If {$\abs{B[v, i]} + \abs{D[v, i]} \leq M$ and $m[v, i] = 0$ \nllabel{alg:computecond}} {
			$\algcompute(v, i)$\;
		}
		\lIf {$i = h$} {
            $d[v] \define \algestimate(v)$
        }

		\If {$i < h$ and $(Y \neq \emptyset$ or $k[v, i] \neq k)$} {
		\ForEach{$w \in A(v)$, $w \neq u$} {
			$X_1 \define $ pointers of $Y$ in $B[w, i + 1]$\;
			$X_2 \define $ pointers of $Y$ in $D[w, i + 1]$\;
			$\algdelete(B[w, i + 1], X_1)$\;
			$\algdelete(D[w, i + 1], X_2)$\;
			$Z \define$ nodes removed from $B[w, i + 1]$\nllabel{alg:sety}\;

			\If {$k[v, i] \neq k$ and $k = k[w, i + 1]$} {
				$m[w, i + 1] \define m[w, i + 1] - 1$\;
			}
			add $(w, Z)$ to $W$\;

		}
		}
	}
	$U \define W$\
}
delete $u$ from $G$\nllabel{alg:removeu}\;

\end{algorithm}

\begin{proposition}
\label{prop:time}
Assume a graph $G$ with $n$ nodes and $m$ edges.
Assume $0 < \epsilon \leq 1/2$, constant $C$, and the maximum path length $h$.
The running time of \algcore is bounded by
\[
	\bigO{hmM\log \frac{n}{M}} =
	\bigO{hmC\epsilon^{-2}\log \frac{n\epsilon^{2}}{C}}
\]
with a probability of $1 - n\exp(-C)$, where $M$ is defined in Eq.~\ref{eq:buffer}.
\end{proposition}

\begin{proof}
We will prove the proposition by bounding $R_1 + R_2$, where $R_1$ is the
total time needed by \algcompute and the $R_2$ is the total time needed by the
inner loop in \algupdate.

We will bound $R_1$ first.
Note that a single call of $\algcompute(v, i)$ requires $\bigO{M\deg v}$ time.

To bound the number of \algcompute calls,
let us first bound $k[v, i]$. 
Proposition~\ref{prop:approx} and union bound implies that
\[
	M2^{k[v, i] - 1} \leq (1 + \epsilon)c(v) \leq 2n
\]
holds for all nodes $v \in V$ with probability $1 - n\exp(-C)$.
Solving for $k[v, i]$ leads to
\begin{equation}
\label{eq:levelbound}
	k[v, i] \leq 2 + \log_2 \frac{n}{M} \in \bigO{\log \frac{n}{M}} \quad.
\end{equation}

We claim that $\algcompute(v, i)$ is called at most twice per each value of $k[v, i]$.
To see this, consider that $\algcompute(v, i)$ sets $m[v, i] = 0$. Then we also set $D[v, i]$ and we are guaranteed by
the first condition on Line~\ref{alg:computecond} of \algupdate that the next call of $\algcompute(v, i)$ will lower $k[v, i]$.
Assume now that $\algcompute(v, i)$ sets $m[v, i] > 0$.
Then the second condition on Line~\ref{alg:computecond} of \algupdate guarantees that the next call of $\algcompute(v, i)$ 
either keeps $m[v, i]$ at 0 (and computes $D[v, i]$) or lowers $k[v, i]$.

In summary, the time needed by \algcompute is bounded by 
\[
	R_1 \in \bigO{\sum_{i, v} M\deg(v) \log \frac{n}{M}}
	= \bigO{hmM\log \frac{n}{M}} \quad.
\]

Let us now consider $R_2$.
For each deleted node in $B[v, i]$ or for each lowered $k[v, i]$ the inner for-loop requires $\bigO{\deg v}$ time.
Equation~\ref{eq:levelbound} implies that the total number of deletions from $B[v, i]$ is in $\bigO{M\log \frac{n}{M}}$,
and that we can lower $k[v, i]$ at most $\bigO{\log \frac{n}{M}}$ times.
Consequently,
\[
	R_2 \in \bigO{h\sum_v (M + 1)\log \frac{n}{M} \deg v} = \bigO{hmM\log \frac{n}{M} }\ .
\]
We have bounded both $R_1$ and $R_2$ proving the main claim. 
\end{proof}

\begin{corollary}
\label{cor:time}
Assume real values $\epsilon > 0$, $\delta > 0$, a graph $G$ with $n$ nodes and $m$ edges.
Let $C = \log(2n / \delta)$. Then \algcore yields $\epsilon$ approximation in
\[
	\bigO{\frac{hm\log n / \delta}{\epsilon^{2}}\log \frac{n\epsilon^{2}}{\log n / \delta}}
\]
time with $1 - \delta$ probability.
\end{corollary}

\begin{proposition}
\algcore requires $\bigO{hmM}$ memory.
\end{proposition}

\begin{proof}
An entry in $B[v, i]$ requires $\bigO{\deg v}$
memory for the pointer information. An entry in $D[v, i]$ only requires $\bigO{1}$ memory.
Since $\abs{B[v, i]} \leq M$ and $\abs{D[v, i]} \leq M \deg v$, the claim follows.
\end{proof}

In order to speed-up the algorithm further we employ two additional heuristics.
First, we can safely delay the initialization of $B[v, i]$ until every $B[w, i - 1]$, where $w \in A(v)$, yields a core
estimate that is below the current core number.
Delaying the initialization allows us to ignore $B[v, i]$ during \algupdate.
Second, if the current core number exceeds the number of remaining nodes, then we can
stop and use the current core number for the remaining nodes. While these heuristics do not provide
any additional theoretical advantage, they speed-up the algorithm in practice.

\section{Distance-generalized dense subgraphs}
\label{sec:dense}

In this section we will study distance-generalized dense subgraphs, a notion introduced by~\citet{bonchi2019distance}.

In order to define the problem,
let us first define $E_h(X)$ to be the node pairs in $X$ that are
connected with an $h$-path in $X$. We exclude the node pairs of form $(u, u)$.
Note that $E(X) = E_1(X)$.

We define the $h$-density of $X$ to be the ratio of $E_h(X)$ and $\abs{X}$,
\[
	\dens{X; h} = \frac{\abs{E_h(X)}}{\abs{X}}\quad.
\]
We will sometimes drop $h$ from the notation if it is clear from the context.

\begin{problem}[\prbdense]
Given a graph $G$ and $h$ find the subgraph $X$ maximizing $\dens{X; h}$.
\end{problem}

\prbdense can be solved for $h = 1$ in polynomial time using fractional
programming combined with minimum cut problems~\citep{Goldberg:1984up}. However, the
distance-generalized version of the problem is \np-hard.

\begin{proposition}
\label{prop:densenp}
\prbdense is \np-hard even for $h = 2$.
\end{proposition}

To prove the result we will use extensively the following lemma.

\begin{lemma}
\label{lem:dens}
Let $X$ be the densest subgraph. Let $Y \subseteq X$ and $Z \cap X = \emptyset$. Then
\[
	\frac{\abs{E_h(X)} - \abs{E_h(X \setminus Y)}}{\abs{Y}} \geq \dens{X} \geq \frac{\abs{E_h(X \cup Z)} - \abs{E_h(X)}}{\abs{Z}}
\]
\end{lemma}

\begin{proof}
Due to optimality $\dens{X} \geq \dens{X \setminus Y}$. Then
\[
	\frac{\abs{E_h(X)} - \abs{E_h(X \setminus Y)}}{\abs{Y}} \geq \frac{\abs{E_h(X)} - \dens{X} (\abs{X} - \abs{Y})}{\abs{X} - (\abs{X} - \abs{Y})}  = \dens{X}\quad.
\]
Similarly, $\dens{X} \geq \dens{X \cup Z}$ implies
\[
	\frac{\abs{E_h(X \cup Z)} - \abs{E_h(X)}}{\abs{Z}} \leq \frac{\dens{X}(\abs{X} + \abs{Z}) - \abs{E_h(X)}}{(\abs{X} + \abs{Z}) - \abs{X}} = \dens{X},
\]
proving the claim.
\end{proof}

\begin{proof}[Proof of Proposition~\ref{prop:densenp}]
To prove the claim we will reduce \prbmatch to our problem.
In an instance of \prbmatch we are given
a universe $U = u_1, \ldots, u_n$ of size $n$ and $m$ sets $\fm{C}$ of size $3$
and ask whether there is an exact cover of $U$ in $\fm{C}$.

We can safely assume that $C_1$ does not intersect with any other set. Otherwise, we can
add a new set and 3 new items without changing the outcome of the instance.

In order to define the graph, let us first define $k = 12m$ and $\ell = 3k(3k - 1)/2 + 6k - 1$.
Note that $k \geq 12$.

For each $u_i \in U$, we add $k$ nodes $a_{ij}$, where $j = 1, \ldots, k$. For each $a_{ij}$,
we add $2\ell$ unique nodes that are all connected to $a_{ij}$. We will denote the resulting star
with $S_{ij}$. We will select a  non-center node from $S_{ij}$ and denote it by $b_{ij}$.
Finally, we write $S'_{ij} = S_{ij} \setminus \set{a_{ij}, b_{ij}}$.

For each set $C_t \in \fm{C}$, we add a node, say $p_t$, and connect it to $b_{ij}$ for every $u_i \in C$
and $j = 1, \ldots, k$. We will denote the collection of these nodes with $P$.
We connect every node in $P$ to $p_1$.

Let $X$ be the nodes of the densest subgraph for $h = 2$.
Let $Q = P \cap X$ and let $\fm{R}$ be the corresponding sets in $\fm{C}$.

To simplify the notation we will need the following counts of node pairs.
First, let us define $\alpha$ to be the number of node pairs in a single $S_{ij}$
connected with a 2-path,
\[
	\alpha = E_2(S_{ij}) = {2 \ell + 1 \choose 2}\quad.
\]
Second, let us define the number of node pairs connected with a 2-path using a single node $p_t \in P$.
Since $p_t$ connects $3k$ nodes $b_{ij}$ and reaches $3k$ nodes $b_{ij}$ and $3k$ nodes $a_{ij}$, we have
\[
	\beta = {3k \choose 2} + 6k\quad.
\]

Finally, consider $W$ consisting of a single $p_t$ and the corresponding $3k$ stars.
Let us write $\gamma = 3k\alpha + \beta$ to be the number of node pairs 
connected by a 2-path in $W$.

We will prove the proposition with a series of claims.

\emph{Claim 1:} $\dens{X} > \ell$.
The density of $W$ as defined above is
\[
	\dens{W} = \frac{3k \alpha + \beta}{3k(2\ell + 1) + 1} > \frac{3k \alpha + \ell}{3k(2\ell + 1) + 1}
	= \ell\quad.
\]
Thus, $\dens{X} \geq \dens{W} > \ell$.

\emph{Claim 2:} $\fm{R}$ is disjoint.
To prove the claim, assume otherwise and let $C_t$, with $t > 1$, be a set overlapping
with some other set in $\fm{R}$. 

Let us bound the number of node pairs that are \emph{solely} connected with $p_t$. 
The node $p_t$ connects $3k + 1$ nodes in $V$. Out of these nodes at least $k + 1$ nodes are
connected by another node in $X$. In addition, $p_t$ reaches to $a_{ij}$ and $b_{ij}$,
where $u_i \in C_t$ and $j = 1, \ldots, k$, that is, $6k$ nodes in total.
Finally, $p_t$ may connect to every other node in $P$, at most $m - 1$ nodes, and every $a_{ij}$ connected to $p_1$, at most $3k$ nodes.
In summary, we have
\[
\begin{split}
	\abs{E_2(X)} - \abs{E_2(X \setminus \set{p_t})} & \leq {3k + 1 \choose 2} - {k + 1 \choose 2} + 6k + m - 1 + 3k\\
	& = \ell - k^2/2 + 5k/2 + m + 3k 
	< \ell \leq \dens{X}\quad.
\end{split}
\]
Lemma~\ref{lem:dens} with $Y = \set{p_t}$ now contradicts the optimality of $X$.
Thus, $\fm{R}$ is disjoint.

\iffalse
\[
	{3k + 1 \choose 2} - {k + 1 \choose 2} = {2k \choose 2} + 2k(k + 1) = k(2k - 1) + 2k(k + 1) = k(4k + 1)
\]
(3k + 1)3k - (k + 1)k = 9k^2 + 3k - k^2 - k = 8k^2 + 2k 

\[
	{3k \choose 2} - k(4k + 1) = 3k(3k - 1)/2 - k(4k + 1) = 1/2(3k(3k - 1) - 8k^2 - 2k) = k^2/2  - 5k/2
\]
\fi

\emph{Claim 3:} Either $S_{ij} \subseteq X$ or  $S_{ij} \cap X = \emptyset$.
To prove the claim assume that $S_{ij} \cap X \neq \emptyset$.

Assume that $b_{ij} \notin X$. Then $S_{ij} \cap X$ is a disconnected component with density less than $\ell$,
contradicting Lemma~\ref{lem:dens}.
Assume that $b_{ij} \in X$ and $a_{ij} \notin X$. Then deleting $b_{ij}$ will reduce at most $3k + m - 1 < \ell$ connected node pairs,
contradicting Lemma~\ref{lem:dens}.

Assume that $b_{ij}, a_{ij} \in X$.  
If $S'_{ij} \cap X = \emptyset$, then deleting $a_{ij}$ will reduce at most $2$ connected node pairs,
contradicting Lemma~\ref{lem:dens}.
Assume now there are $u \in S'_{ij} \cap X$ and $w \in  S'_{ij} \setminus X$.
Then $\abs{E_2(X \cup \set{w})} - \abs{E_2(X)} > \abs{E_2(X)} - \abs{E_2(X \setminus \set{u})}$,
contradicting Lemma~\ref{lem:dens}.
Consequently, $S_{ij} \subseteq X$.

\emph{Claim 4:} If $p_t \in X$, then $X$ contains every corresponding $S_{ij}$.
To prove the claim assume otherwise.

Assume first that there are no corresponding $S_{ij}$ in $X$ for $p_t$.
If $t > 1$, then $p_t$ reaches to at most $m - 1 + 3k$ nodes.
If $t = 1$, then $p_1$ connects at most $m - 1$ nodes and reaches to at most $(m - 1)(3k + 1)$ nodes.

Both cases lead to
\[
\begin{split}
    \abs{E_2(X)} - \abs{E_2(X \setminus \set{p_t})} & \leq {m - 1 \choose 2} + (m - 1)(3k + 1) 
    < \ell < \dens{X},
\end{split}
\]
contradicting Lemma~\ref{lem:dens}.

Assume there is at least one corresponding $S_{ij}$ in $X$ but not all, say $S_{i'j'}$ is missing.
Then
\[
	\abs{E_2(X)} - \abs{E_2(X \setminus S_{ij})} < \abs{E_2(X \cup S_{i'j'})} - \abs{E_2(X)},
\]
contradicting Lemma~\ref{lem:dens}.

\emph{Claim 5:}
No $S_{ij}$ without corresponding $p_t$ is included in $X$. To prove the claim note
that such $S_{ij}$ is disconnected and has density of $\ell$, 
contradicting Lemma~\ref{lem:dens}.

The previous claims together show that the density of $X$ is equal to
\[
	\dens{X} = \frac{\abs{Q} \gamma + (\abs{Q} - 1)(6k) + {\abs{Q} \choose 2}}{\abs{Q}(3k(2\ell + 1) + 1)},
\]
which is an increasing function of $\abs{Q}$. Since $\fm{R}$ is disjoint and maximal, the \prbmatch instance has a solution
if and only if $\fm{R}$ is a solution.
\end{proof}

One of the appealing aspects of $\dens{X; h }$ for $h = 1$ is that we can
2-approximate in linear time~\citep{Charikar:2000tg}. This is done by ordering the
nodes with \algexact, say $v_1, \ldots, v_n$ and then selecting the densest
subgraph of the form $v_1, \ldots, v_i$.

The approximation guarantee for $h > 1$ is weaker even if use \algexact.
\citet{bonchi2019distance} showed that $2\dens{Y} \geq \sqrt{2\dens{X} + 1/4} - 1/2$
when we use \algexact.

Using \algcore instead of \algexact poses additional challenges.  In order to
select a subgraph among $n$ candidates, we need to estimate the density of its
subgraph. We cannot use $d[v]$ used by \algcore as these are the values that \algcore
uses to determine the order.

Assume that \algcore produced order of vertices $v_1, \ldots, v_n$, first
vertices deleted first. To find the densest graph among the candidate, we
essentially repeat \algcore except now we delete the nodes using the order
$v_1, \ldots, v_n$. We then estimate the number of edges with the identity
\[
	2\abs{E_h(X)} = \sum_{v \in X} \deg_h(v; X)\quad.
\]
We will refer to this algorithm as \algdense. The pseudo-code for \algdense is given in Algorithm~\ref{alg:dense}.

\begin{algorithm}
\caption{$\algdense(G, v_1, \ldots, v_n, \epsilon, C)$ approximative dense subgraph. Setting
$C = \log(n^2 / \delta)$ yields an approximation with $1 - \delta$ probability.}
\label{alg:dense}

\ForEach {$v \in V$} {
    $r[v] \define \geom{1/2}$\;
    $B[v, 0] \define \set{v}$\;
    $k[v, 0] \define 0$\;
}

$M \define 1 + \frac{4(2 + \epsilon)}{\epsilon^2}(C + \log 8)$\;

\ForEach {$i = 1, \ldots, h$} {
    \lForEach {$v \in V$} {
        $\algcompute(v, i)$
    }
}
$R \define \sum_v d[v]$\;

\ForEach {$i = 1,\ldots, n$} {
	estimate $\dens{v_i, \ldots, v_n}$ with $R / (n - i + 1)$\;
    $\algupdate(v_i)$\;
	keep $R$ updated when recomputing $d[v]$ with \algestimate\;
}
\Return densest tested subgraph\;
\end{algorithm}

The algorithm yields to the following guarantee.

\begin{proposition}
\label{prop:dense}
Assume $\epsilon > 0, C > 0$ and $h$. Define $\gamma = \frac{1 - \epsilon}{1 + \epsilon}$.
For any given $k$, let $C_k$ be the $(k, h)$-core. Define 
\[
	\beta = \min_{k} \frac{\abs{C_k}}{\abs{C_{k\gamma}}}
\]
to be the smallest size ratio between $C_k$ and $C_{k\gamma}$.

Let $X$ be the $h$-densest subgraph.

Let $c'$ be an $\epsilon$-approximative core map and let $v_1, \ldots, v_n$ be the corresponding vertex order.
Let $Y = \algdense(G, v_1, \ldots, v_n, \epsilon, C)$
Then
\[
	2\dens{Y} \geq  \gamma \beta \pr{\sqrt{2 \dens{X} + 1/4 } - 1/2}
\]
with probability $1 - n^2\exp\pr{-C}$.
\end{proposition}

To prove the result we need the following lemma.

\begin{lemma}
\label{lem:densebound}
For any given $k$, define $C'_k = \set{v \mid c'(v) \geq k}$. Then
\[
	\dens{C_{k(1 - \epsilon)}'} \geq \beta \dens{C_k}\quad.
\]
\end{lemma}

\begin{proof}
Write $F = C_{k(1 - \epsilon)}'$.
Let $v \in C_k$. Then $c'(v) \geq (1 - \epsilon) c(v) \geq (1 - \epsilon)k$ and so $v \in F$.  Thus $C_k \subseteq F$.
Conversely, let $v \in F$. Then
$(1 + \epsilon)c(v) \geq c'(v) \geq k(1 - \epsilon)$ and so $v \in C_{\gamma k}$.
Thus $F \subseteq C_{\gamma k}$. The definition of $\beta$ now implies
\[
	\dens{F} = \frac{\abs{E_h(F)}}{\abs{F}} \geq \beta \frac{\abs{E_h(F)}}{\abs{C_k}} \geq \beta \frac{\abs{E_h(C_k)}}{\abs{C_k}}
\]
proving the claim.
\end{proof}

\begin{proof}[Proof of Proposition~\ref{prop:dense}]
Let $c$ be the core map produced by \algexact.
For any given $k$, define $C'_k = \set{v \mid c'(v) \geq k}$.

Let $u \in X$ be the first vertex deleted by \algexact.
Let $b = \deg_h(u; X)$ be its $h$-degree.
Write $X' = X \setminus \set{u}$. Since $X$ is optimal,
\[
	\frac{\abs{E_h(X)}}{\abs{X}} \geq
	\frac{\abs{E_h(X')}}{\abs{X'}}\quad.
\]
Deleting $u$ from $X$ will delete $b$ node pairs from $E_h(X)$ containing $u$.
In addition, every node in the $h$-neighborhood of $u$ may be disconnected from each other,
potentially reducing the node pairs by ${b \choose 2}$. In summary,
\[
	\abs{E_h(X)} - \abs{E_h(X')} \leq b + {b \choose 2} = {b + 1 \choose 2}\quad.
\]

Combining the two inequalities leads to
\[
	{b + 1 \choose 2} \geq \abs{E_h(X)} - \frac{\abs{E_h(X)}(\abs{X} - 1)}{\abs{X}} = \frac{\abs{E_h(X)}}{\abs{X}} = \dens{X}\quad.
\]
Solving for $b$ results in 
\begin{equation}
\label{eq:dense1}
	b  \geq \sqrt{2 \dens{X} + 1/4 } - 1/2\quad.
\end{equation}

Let $Z$ be the nodes right before $u$ is deleted by \algexact.
Note that $c(u) \geq \deg_h(u; Z) \geq \deg_h(u; X) = b$.

Let $C_k$ be the smallest core containing $u$, that is, $c(u) = k$.
By definition, $\deg_h(v; C_k) \geq k \geq b$,  for all $v \in C_k$.

Let $F = C'_{k(1 - \epsilon)}$.
Lemma~\ref{lem:densebound} now states that
\begin{equation}
\label{eq:dense2}
	2\dens{F} \geq 2 \beta \dens{C_k} = \beta \frac{1}{\abs{C_k}} \sum_{v \in C_k} \deg_h(v; C_k) \geq \beta k \geq \beta b\quad.
\end{equation}

Let $d'(Z)$ be the estimated density for a subgraph $Z$.

Proposition~\ref{prop:approx} and the union bound state that
\begin{equation}
\label{eq:dense3}
	\dens{Y} \geq \frac{1}{1 + \epsilon} d'(Y) \geq \frac{1}{1 + \epsilon} d'(F) \geq \gamma \dens{F}
\end{equation}
with probability $1 - n^2e^{-C}$.
Eqs.~\ref{eq:dense1}--\ref{eq:dense3} prove the inequality in the claim.
\end{proof}

\algdense is essentially \algcore so we can apply Proposition~\ref{prop:time}.

\begin{corollary}
Assume real values $\epsilon > 0$, $\delta > 0$, a graph $G$ with $n$ nodes and $m$ edges.
Let $C = \log(n^2 / \delta)$. Then
\algdense
runs in
\[
    \bigO{\frac{hm\log n / \delta}{\epsilon^{2}}\log \frac{n\epsilon^{2}}{\log n / \delta}}
\]
time and Proposition~\ref{prop:dense} holds with $1 - \delta$ probability.
\end{corollary}

Finally, let us describe a potentially faster variant of the algorithm that we will use in our experiments.
The above proof will work even if replace $C_k$ with the most inner (exact) core.
Since $F = C'_{k(1 - \epsilon)}$ we can prune all the vertices for which $c'(v) < k(1 - \epsilon)$.
The problem is that we do not know $k$ but we can lower bound it with $k \geq k'/(1 + \epsilon)$,
where $k' = \max_v c'(v)$. In summary, before running \algestimate we remove all the vertices for
which $c'(v) < \gamma k'$.

\section{Related work}\label{sec:related}

The notion of distance-generalized core decomposition was proposed
by~\citet{bonchi2019distance}. The authors provide several heuristics to
significantly speed-up the baseline algorithm (a variant of an algorithm
proposed by~\citet{batagelj2011fast}).
Despite being significantly faster than
the baseline approach, these heuristics still have the computational complexity
in $\bigO{nn'(n' + m')}$, where $n'$ and $m'$ are the numbers of nodes and
edges in the largest $h$-neighborhood. For dense graphs and large values of
$h$, the sizes $n'$ and $m'$ can be close $n$ and $m$, leading to the
computational time of $\bigO{n^2m}$. We will use these heuristics as baselines
in Section~\ref{sec:exp}.

All these algorithms, as well as ours, rely on the same idea of iteratively
deleting the vertex with the smallest $\deg_h(v)$ and updating these
counters upon the deletion. The difference is that the previous works maintain
these counters exactly---and use some heuristics to avoid updating unnecessary
nodes---whereas we approximate $\deg_h(v)$ by sampling, thus reducing the
computational time at the cost of accuracy.

A popular variant of decomposition is a $k$-truss, where each edge is required
to be in at least $k$ triangles~\citep{huang2014querying,zhao2012large,zhang2012extracting,wang2012truss,cohen2008trusses}.
\citet{sariyuce2015finding,sariyuce2016fast} proposed $(r, s)$ nucleus decomposition, an extension of $k$-cores where the notion
nodes and edges are replaced with $r$-cliques and $s$-cliques, respectively. 
\citet{sariyuce2016fast} points out that there are several variants of $k$-trusses, depending 
on the connectivity requirements: \citet{huang2014querying} requires the trusses to be triangle-connected, \citet{cohen2008trusses}
requires them to be connected, and \citet{zhang2012extracting} allows the trusses to be disconnected.

A $k$-core is the largest subgraph whose smallest degree is at least $k$.  A
similar concept is the densest subgraph, a subgraph whose average degree is the
largest~\citep{Goldberg:1984up}. Such graphs are convenient variants for
discovering dense communities as they can be discovered in polynomial time~\citep{Goldberg:1984up},
as opposed to, e.g., cliques that are inapproximable~\citep{zuckerman2006linear}.

Interestingly, the same peeling algorithm that is used for core decomposition
can be use to 2-approximate the densest subgraph~\citep{Charikar:2000tg}.
\citet{tatti2019density} proposed a variant of core decomposition so that the densest subgraph
is equal to the inner core. This composition is solvable in polynomial time
an can be approximated using the same peeling strategy.

A distance-generalized clique is known as $h$-club, which is a subgraph
where every node is reachable by an $h$-path from every node~\citep{mokken1979cliques}.
Here the path must be inside the subgraph.
Since cliques are
$1$-clubs, discovering maximum $h$-clubs is immediately an inapproximable problem.
\citet{bonchi2019distance} argued that $(k, h)$ decomposition can be used to
aid discovering maximum $h$-clubs.

Using sampling for parallelizing (normal) core computation was proposed
by~\citet{esfandiari2018parallel}. Here, the authors sparsify the graph
multiple times by sampling edges. The sampling probability depends on the core
numbers: larger core numbers allow for more aggressive sparsification. 
The authors then use Chernoff bounds to prove the approximation guarantees.
The authors were able to sample edges since the degree in the
sparsified graph is an estimate of the degree in the original graph (multiplied by the sampling probability).
This does not hold for $(k, h)$ core decomposition 
because a node $w \in N(v; h)$ can reach $v$ with several paths.

Approximating $h$-neighborhoods can be seen as an instance of a cardinality estimation problem.
A classic approach for solving such problems is HyperLogLog~\citep{flajolet2007hyperloglog}.
Adopting HyperLogLog or an alternative approach, such as \citep{kane2010optimal}, 
is a promising direction for a future work, potentially speeding up the algorithm further. 
The challenge here is to maintain the estimates as the nodes are removed by \algcore.

\begin{table*}[t]
\caption{Sizes and computational times for the benchmark datasets. Here, $n$ is
the number of nodes $m$ is the number of edges, $M$ is the internal parameter
of \algcore given in Eq.~\ref{eq:buffer}. The running times for the baselines \alglub and \alglb
are taken from~\citep{bonchi2019distance}.
Dashes indicate that the experiments did not finish in 20 hours.
For \dtname{Youtube} and \dtname{Hyves}, \alglub was run with 52 CPU cores.
The remaining experiments are done with a single CPU core.
}
\label{tab:stats}

\begin{filecontents*}{stats.csv}
name,n,m,M,h1raw,h2raw,h3raw,b1raw,b2v,b3raw,c1raw,c2raw,c3raw,h1,h2,h3,b1,b2,b3,c1,c2,c3
\dtname{CaHep},12008,118489,607,3.53,10.39,22.34,0.95,128.16,940.69,1.19,92.68,122.54,3.53s,10.39s,22.34s,0.95s,2m8s,15m41s,1.19s,1m33s,2m3s
\dtname{CaAstro},18772,198050,625,5.2,23.65,20.81,5.52,560.2,4835.06,5.17,91.39,372.93,5.2s,23.65s,20.81s,5.52s,9m20s,80m35s,5.17s,1m31s,6m13s
\dtname{RoadPA},1088092,1541898,787,7.72,13.85,23.64,3.18,6.75,11.47,36.14,118.94,139.8,7.72s,13.85s,23.64s,3.18s,6.75s,11.47s,36.14s,1m59s,2m20s
\dtname{RoadTX},1393383,1921660,797,10.39,18.28,30.9,4.21,8.44,13.9,56.89,184.29,208.38,10.39s,18.28s,30.9s,4.21s,8.44s,13.9s,56.89s,3m4s,3m28s
\dtname{Amazon},334863,925872,740,4.96,17.26,75.21,2.51,29.27,295.78,12.98,51.92,190.88,4.96s,17.26s,1m15s,2.51s,29.27s,4m56s,12.98s,51.92s,3m11s
\dtname{Douban},154908,327162,709,6.39,57.78,93.75,4.3,1864.09,54762.1,6.76,220.72,3556.72,6.39s,57.78s,1m34s,4.3s,31m4s,912m42s,6.76s,3m41s,59m17s
\dtname{Hyves},1402673,2777419,797,81.91,426.02,723.35,113.48,42163.60,,440.93,3724.94,48038.7,1m22s,7m6s,12m3s,1m53s,702m44s,---,7m21s,62m5s,800m39s
\dtname{Youtube},495957,1936748,756,71.18,311.72,265.55,102.75,,,192.46,3192.07,9310.85,1m11s,5m12s,4m26s,1m43s,---,---,3m12s,53m12s,155m11s
\end{filecontents*}
\pgfplotstabletypeset[
    begin table={\begin{tabular*}{\textwidth}},
    end table={\end{tabular*}},
    col sep=comma,
	columns = {name, n, m, M, h1, b1, c1},
    columns/name/.style={string type, column type={@{\extracolsep{\fill}}l}, column name=\emph{Dataset}},
    columns/n/.style={fixed, set thousands separator={\,}, column type=r, column name=$n$},
    columns/m/.style={fixed, set thousands separator={\,}, column type=r, column name=$m$},
    columns/M/.style={fixed, set thousands separator={\,}, column type=r, column name=$M$},
    columns/b1/.style={column type=r, column name={\alglb}, string type},
    columns/b2/.style={column type=r, column name={\alglb}, string type},
    columns/b3/.style={column type=r, column name={\alglb}, string type},
    columns/c1/.style={column type=r, column name={\alglub}, string type},
    columns/c2/.style={column type=r, column name={\alglub}, string type},
    columns/c3/.style={column type=r, column name={\alglub}, string type},
    columns/h1/.style={column type=r, column name={\algcore}, string type},
    columns/h2/.style={column type=r, column name={\algcore}, string type},
    columns/h3/.style={column type=r, column name={\algcore}, string type},
    every head row/.style={
		before row={
			\toprule
			\arraybackslash
			& & & & \multicolumn{3}{l}{$h = 2$} \\
			\cmidrule(lr){5-7}},
			after row=\midrule},
    every last row/.style={after row=\bottomrule},
]
{stats.csv}

\pgfplotstabletypeset[
    begin table={\begin{tabular*}{\textwidth}},
    end table={\end{tabular*}},
    col sep=comma,
	columns = {name, h2, b2, c2, h3, b3, c3},
    columns/name/.style={string type, column type={@{\extracolsep{\fill}}l}, column name=\emph{Dataset}},
    columns/n/.style={fixed, set thousands separator={\,}, column type=r, column name=$n$},
    columns/m/.style={fixed, set thousands separator={\,}, column type=r, column name=$m$},
    columns/M/.style={fixed, set thousands separator={\,}, column type=r, column name=$M$},
    columns/b1/.style={column type=r, column name={\alglb}, string type},
    columns/b2/.style={column type=r, column name={\alglb}, string type},
    columns/b3/.style={column type=r, column name={\alglb}, string type},
    columns/c1/.style={column type=r, column name={\alglub}, string type},
    columns/c2/.style={column type=r, column name={\alglub}, string type},
    columns/c3/.style={column type=r, column name={\alglub}, string type},
    columns/h1/.style={column type=r, column name={\algcore}, string type},
    columns/h2/.style={column type=r, column name={\algcore}, string type},
    columns/h3/.style={column type=r, column name={\algcore}, string type},
    every head row/.style={
		before row={
			\toprule
			\arraybackslash
			& \multicolumn{3}{l}{$h = 3$} & \multicolumn{3}{l}{$h = 4$} \\
			\cmidrule(lr){2-4}
			\cmidrule(lr){5-7}
			},
			after row=\midrule},
    every last row/.style={after row=\bottomrule},
]
{stats.csv}
\end{table*}

\section{Experimental evaluation}\label{sec:exp}

Our two main goals in experimental evaluation is to study the accuracy and the
computational time of \algcore.

\subsection{Datasets and setup}
We used 8 publicly available benchmark datasets.
\dtname{CaAstro} and \dtname{CaHep} are collaboration networks between researchers.\!\footnote{\url{http://snap.stanford.edu}\label{foot:snap}}
\dtname{RoadPa} and \dtname{RoadTX} are road networks in Pennsylvania and Texas.\!\footref{foot:snap}
\dtname{Amazon} contains product pairs that are often co-purchased in a popular online retailer.\!\footref{foot:snap}
\dtname{Youtube} contains user-to-user links in a popular video streaming service.\!\footnote{\url{http://networkrepository.com/}}
\dtname{Hyves} and \dtname{Douban} contain friendship links in a Dutch and Chinese social networks, respectively.\!\footnote{\url{http://konect.cc/}}
The sizes of the graphs are given in Table~\ref{tab:stats}.

We implemented \algcore in C++\footnote{\url{http://version.helsinki.fi/dacs}}
and conducted the experiments using a single core (2.4GHz).
For \algcore we used 8GB RAM and for \algdense we used 50GB RAM.
In all experiments, we set $\delta = 0.05$.

\subsection{Accuracy}

In our first experiment we compared the accuracy of our estimate $c'(v)$ against the correct core numbers $c(v)$.
As a measure we used the maximum relative error
\[
	\max_{v \in V} \frac{\abs{c'(v) - c(v)}}{c(v)}\quad.
\]
Note that Proposition~\ref{prop:yield} states that the error should be less than $\epsilon$ with high probability.

The error as a function of $\epsilon$ for \dtname{CaHep} and \dtname{CaAstro} datasets is shown in Figure~\ref{fig:accuracy} for $h = 3, 4$.
We see from the results that the error tends to increase as a function of $\epsilon$.
As $\epsilon$ decreases, the internal value $M$ increases, reaching the point where the maximum core number is
smaller than $M$. For such values, Proposition~\ref{prop:yield} guarantees that \algcore produces correct results.
We see, for example, that this value is reached with $\epsilon = 0.20$ for \dtname{CaHep},
and $\epsilon = 0.15$ for \dtname{CaAstro} when $h = 3$, and $\epsilon = 0.35$ for \dtname{Amazon} when $h = 4$.

\begin{figure}[t!]
\begin{tabular}{ll}
\begin{tikzpicture}
\begin{axis}[xlabel={$\epsilon$}, ylabel= {relative error},
    width = 5.7cm,
    xmin = 0.1,
    xmax = 0.5,
    ymax = 0.15,
    scaled y ticks = false,
    cycle list name=yaf,
    yticklabel style={/pgf/number format/fixed},
    no markers,
    ]
\addplot table[x = eps, y = cahea] {accuracy.csv} node[pos=0.58, sloped, anchor=north, text=black, inner sep = 1pt] {\dtname{CaHep}};
\addplot table[x = eps, y = caasa] {accuracy.csv} node[pos=0.32, sloped, anchor=south, text=black, inner sep = 1pt] {\dtname{CaAstro}};
\addplot table[x = eps, y = amza] {accuracy.csv} node[pos=0.75, sloped, anchor=south, text=black, inner sep = 1pt] {\dtname{Amazon}};
\pgfplotsextra{\yafdrawaxis{0.1}{0.5}{0}{0.15}}
\end{axis}
\end{tikzpicture}&
\begin{tikzpicture}
\begin{axis}[xlabel={$\epsilon$}, ylabel= {time (seconds)},
    width = 5.7cm,
    xmin = 0.1,
    xmax = 0.5,
	ymin = 0,
    scaled y ticks = false,
    cycle list name=yaf,
    yticklabel style={/pgf/number format/fixed},
    no markers,
    ]
\addplot table[x = eps, y = cahet] {accuracy.csv} node[pos=0.02, sloped, anchor=south, text=black, inner sep = -1pt] {\dtname{CaHep}};
\addplot table[x = eps, y = caast] {accuracy.csv} node[pos=0.55, sloped, anchor=south, text=black, inner sep = 1pt] {\dtname{CaAstro}};
\addplot table[x = eps, y = amzt] {accuracy.csv} node[pos=0.10, sloped, anchor=north, text=black, inner sep = 1pt] {\dtname{Amazon}};
\pgfplotsextra{\yafdrawaxis{0.1}{0.5}{0}{95}}
\end{axis}
\end{tikzpicture}\\
\begin{tikzpicture}
\begin{axis}[xlabel={$\epsilon$}, ylabel= {relative error},
    width = 5.7cm,
    xmin = 0.1,
    xmax = 0.5,
    ymax = 0.2,
    scaled y ticks = false,
    cycle list name=yaf,
    yticklabel style={/pgf/number format/fixed},
    no markers,
    ]
\addplot table[x = eps, y = cahea] {accuracy2.csv} node[pos=0.44, sloped, anchor=north, text=black, inner sep = 1pt] {\dtname{CaHep}};
\addplot table[x = eps, y = caasa] {accuracy2.csv} node[pos=0.32, sloped, anchor=south, text=black, inner sep = 1pt] {\dtname{CaAstro}};
\addplot table[x = eps, y = amza] {accuracy2.csv} node[pos=0.4, sloped, anchor=south, text=black, inner sep = 1pt] {\dtname{Amazon}};
\pgfplotsextra{\yafdrawaxis{0.1}{0.5}{0}{0.20}}
\end{axis}
\end{tikzpicture}&
\begin{tikzpicture}
\begin{axis}[xlabel={$\epsilon$}, ylabel= {time (seconds)},
    width = 5.7cm,
    xmin = 0.1,
    xmax = 0.5,
	ymin = 0,
    scaled y ticks = false,
    cycle list name=yaf,
    yticklabel style={/pgf/number format/fixed},
    no markers,
    ]
\addplot table[x = eps, y = cahet] {accuracy2.csv} node[pos=0.55, sloped, anchor=north, text=black, inner sep = -1pt] {\dtname{CaHep}};
\addplot table[x = eps, y = caast] {accuracy2.csv} node[pos=0.3, sloped, anchor=south, text=black, inner sep = 1pt] {\dtname{CaAstro}};
\addplot table[x = eps, y = amzt] {accuracy2.csv} node[pos=0.50, sloped, anchor=south, text=black, inner sep = 1pt] {\dtname{Amazon}};
\pgfplotsextra{\yafdrawaxis{0.1}{0.5}{0}{180}}
\end{axis}
\end{tikzpicture}
\end{tabular}

\caption{Relative error and computational time as a function of $\epsilon$ for
\dtname{CaHep}, \dtname{CaAstro}, and \dtname{Amazon} datasets and $h = 3$ (top row) and $h = 4$ (bottom row).}

\label{fig:accuracy}
\end{figure}
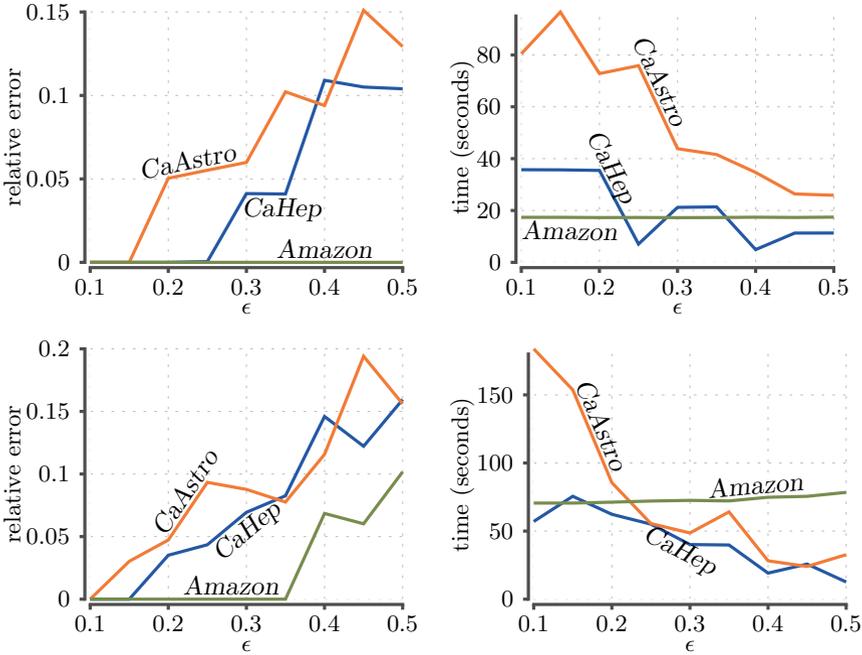

\subsection{Computational time}

Our next experiment is to study the computational time as a function of $\epsilon$; the results are shown in Figure~\ref{fig:accuracy}.
From the results we see that generally computational time increases as $\epsilon$ decreases. The computational time flattens when we reach
the point when $c(v) \leq M$ for every $M$. In such case, the lists $B[v, i]$ match exactly to the neighborhoods $N(v, i)$ and do
not change if $M$ is increased further. Consequently, decreasing $\epsilon$ further will not change the running time. Interestingly,
the running time increases slightly for \dtname{Amazon} and $h = 4$ as $\epsilon$ increases.
This is most likely due to the increased number of \algcompute calls for smaller values of $M$.

Next, we compare the computational time of our method against the baselines \alglb and \alglub proposed by~\citet{bonchi2019distance}.
As our hardware setup is similar, we used the running times for the baselines reported by~\citet{bonchi2019distance}.
Here, we fixed $\epsilon = 0.5$. The results are shown in Table~\ref{tab:stats}.

We see from the results that for $h = 2$ the results are dominated by \alglb. This is due to the fact
that most, if not all, nodes will have $c(v) \leq M$. In such case, \algcore does not use any sampling and does not provide
any speed up. This is especially the case for the road networks, where the core number stays low even for larger values of $h$.
On the other hand, \algcore outperforms the baselines in cases where $c(v)$ is large, whether due to a larger $h$ or due to denser
networks. As an extreme example, \alglub required over 13 hours with 52 CPU cores to compute core for \dtname{Hyves} while \algcore
provided an estimate in about 12 minutes using only 1 CPU core.

Interestingly enough, \algcore solves \dtname{CaAstro} faster when $h = 4$ than when $h = 3$. This is due to the fact that we stop
when the current core value plus one is equal to the number of remaining nodes.

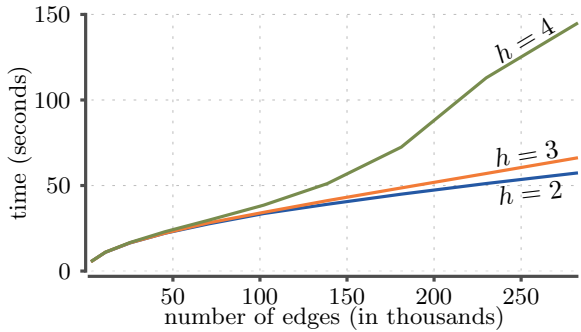
\begin{figure}[ht!]
\begin{center}
\begin{tikzpicture}
\begin{axis}[xlabel={number of edges (in thousands)}, ylabel= {time (seconds)},
    width = 8cm,
    height = 5cm,
	ymin = 0,
	ymax = 150,
    scaled y ticks = false,
    scaled x ticks = false,
    cycle list name=yaf,
    yticklabel style={/pgf/number format/fixed},
    no markers,
    ]
\addplot table[x expr = {\thisrow{m} / 1000}, y = h2] {synth.csv} node[pos=0.90, sloped, anchor=north, text=black, inner sep = 1pt] {$h=2$};
\addplot table[x expr = {\thisrow{m} / 1000}, y = h3] {synth.csv} node[pos=0.90, sloped, anchor=south, text=black, inner sep = 1pt] {$h=3$};
\addplot table[x expr = {\thisrow{m} / 1000}, y = h4] {synth.csv} node[pos=0.90, sloped, anchor=south, text=black, inner sep = 1pt] {$h=4$};
\pgfplotsextra{\yafdrawaxis{2}{282}{0}{150}}
\end{axis}
\end{tikzpicture}
\end{center}
\caption{Computational time as a function of number of edges applied to synthetic data.}
\label{fig:synth}
\end{figure}

To further demonstrate the effect of the network size on the computation time
we generate a series of synthetic datasets.  Each dataset is stochastic
blockmodel with 10 blocks of equal size, $C_1, \ldots, C_{10}$.
To add a hierarchical structure we set the
probability of an edge between nodes in $C_i$ and $C_j$ with $i < j$ to be
$10^{-6} i^2$.
We vary the number of nodes from $10\,000$ to $100\,000$.
The computational times for our method, with $h = 2, 3, 4$ and $\epsilon = 0.5$,
are shown in Figure~\ref{fig:synth}.
As expected, the running times increase as the number of edges increase. Moreover, larger $h$ require more processing time.
We should stress that while Corollary~\ref{cor:time} bounds the running time as quasi-linear,
in practice the trend depends on the underlying model.

\subsection{Dense subgraphs}

\begin{table*}[t]
\caption{Densities and sizes of discovered dense subgraphs for the benchmark datasets.}
\label{tab:dense}

\begin{filecontents*}{dense.csv}
name,d2,n2,d3,n3,d4,n4
\dtname{CaHep}, 494.765727, 1383, 1372.137319, 3998, 3121.324940, 7069
\dtname{CaAstro}, 570.552908, 3525, 2955.164228, 10321, 6280.639205, 15100
\dtname{RoadPA}, 5.184397, 19496, 10.743590, 4407, 18.936488, 15556
\dtname{RoadTX}, 6.000000, 65, 10.599119, 908, 18.281055, 19530
\dtname{Amazon}, 274.500000, 550, 407.248631, 2192, 851.962538, 22476
\dtname{Douban}, 384.687878, 4133, 3435.750090, 13853, 17142.939545, 73840
\dtname{Hyves}, 15135.525279, 31884, 31832.312810, 224136, 142173.061450, 448330
\dtname{Youtube}, 13572.561130, 25413, 44338.226790, 162379, 132376.104336, 315211

\end{filecontents*}
\pgfplotstabletypeset[
    begin table={\begin{tabular*}{\textwidth}},
    end table={\end{tabular*}},
    col sep=comma,
	columns = {name, d2, n2, d3, n3, d4, n4},
    columns/name/.style={string type, column type={@{\extracolsep{\fill}}l}, column name=\emph{Dataset}},
    columns/d2/.style={fixed, set thousands separator={\,}, column type=r, column name=$\dens{X}$},
    columns/d3/.style={fixed, set thousands separator={\,}, column type=r, column name=$\dens{X}$},
    columns/d4/.style={fixed, set thousands separator={\,}, column type=r, column name=$\dens{X}$},
    columns/n2/.style={fixed, set thousands separator={\,}, column type=r, column name=$\abs{X}$},
    columns/n3/.style={fixed, set thousands separator={\,}, column type=r, column name=$\abs{X}$},
    columns/n4/.style={fixed, set thousands separator={\,}, column type=r, column name=$\abs{X}$},
    every head row/.style={
		before row={
			\toprule
			\arraybackslash
			& \multicolumn{2}{l}{$h = 2$} & \multicolumn{2}{l}{$h = 3$} & \multicolumn{2}{l}{$h = 4$} \\
			\cmidrule(lr){2-3}
			\cmidrule(lr){4-5}
			\cmidrule(lr){6-7}
			},
			after row=\midrule},
    every last row/.style={after row=\bottomrule},
]
{dense.csv}
\end{table*}

Finally, we used \algdense to estimate the densest subgraph for $h = 2, 3, 4$.
We set $\epsilon = 0.5$ and $\delta = 0.05$. The results, shown in
Table~\ref{tab:dense}, are as expected.  Both the density and the size of the
$h$-densest subgraphs increase as the function of $h$.  The dense subgraphs are
generally smaller and less dense for the sparse graphs, such as, road networks.

In our experiments, the running times for \algdense were generally smaller but comparable
to the running times of \algcore. The speed-up is largely due to the pruning of nodes
with smaller core numbers. The exception was \dtname{Youtube} with $h = 3$, where \algdense
required over 23 minutes. The slowdown is due to \algcore using lazy initialization of $B[v, i]$
whereas \algdense needs $B[v, h]$ to be computed in order to obtain $d[v]$. This is also the reason
why \algdense requires more memory in practice.

%Even though the computational complexity of \algdense is the same \algcore, \algdense can be
%slower due to larger $M$ and because we cannot use lazy initialization of $B[v, i]$.

\section{Concluding remarks}
\label{sec:conclusions}

In this paper we introduced a randomized algorithm for approximating
distance-generalized core decomposition. The major advantage over the exact approximation
is that the approximation can be done in $\bigO{\epsilon^{-2} hm (\log^2 n - \log \delta)}$ time,
whereas the exact computation may require $\bigO{n^2m}$ time. We also studied distance-generalized dense subgraphs  
by proving that the problem is \np-hard and extended the guarantee results of \citep{bonchi2019distance} to approximate
core decompositions.

The algorithm is based on sampling the $h$-neighborhoods of the nodes. We prove the approximation
guarantee with Chernoff bounds. Maintaining the sampled $h$-neighborhood requires carefully designed
bookkeeping in order to obtain the needed computational complexity. This is especially the case
since the sampling probability may change as the graph gets smaller during the decomposition.

In practice, the sampling complements the exact algorithm. For the setups
where the exact algorithm struggles, our algorithm outperforms the exact approached by a large
margin. Such setups include well-connected networks and values $h$ larger than 3.

An interesting direction for future work is to study whether the heuristics introduced by~\citet{bonchi2019distance}
can be incorporated with the sampling approach in order to obtain even faster decomposition method.

\paragraph{Acknowledgments.}
This research is supported by the Academy of Finland project MALSOME (343045).

\appendix
\section{Proof of Proposition~\ref{prop:approx}}

We start with stating several Chernoff bounds. 

\begin{lemma}[Chernoff bounds]
Let $X_1, \ldots, X_d$ be $d$ independent Bernoulli random variables with $P(X_i = 1) = \mu$.
Let $S = \sum_{i = 1}^d X_i$.
Then
\begin{align}
	P(S \geq (1 + \epsilon) d\mu) & < \exp\pr{-\frac{\epsilon^2}{2 + \epsilon} d \mu}, \label{eq:ch1} \\
	P(S \leq (1 - \epsilon) d\mu) & < \exp\pr{-\frac{\epsilon^2}{2} d \mu}, \quad\text{and}  \label{eq:ch2} \\
	P(\abs{S - d \mu} \geq \epsilon d\mu) & < 2\exp\pr{-\frac{\epsilon^2}{2 + \epsilon} d \mu}\quad . \label{eq:ch3}  
\end{align}
\end{lemma}

\begin{proof}
Eqs.~\ref{eq:ch1}--\ref{eq:ch2} 
are standard multiplicative Chernoff bounds. Eq.~\ref{eq:ch3} 
is obtained with a union bound of Eqs.~\ref{eq:ch1}--\ref{eq:ch2},
completing the claim.
\end{proof}

To prove Proposition~\ref{prop:approx} we first need the following technical lemma.

\begin{lemma}
\label{lem:nice}
Assume $0 < \epsilon \leq 1/2$.
Let $R_1, \ldots, R_d$ be independent random variables sampled from geometric distribution, $\geom{1/2}$.
Define
\[
	S_i = \abs{\set{j \in \spr{d} \mid R_j \geq i}} \text{ and }
	T_i = \abs{\set{j \in \spr{d} \mid R_j \geq i, j \geq 2}}
\]
to be the number of variables $\set{R_j}$ larger than or equal to $i$.
Assume $C > 0$ and define $M$ as in Eq.~\ref{eq:buffer}.
Assume that $M \leq d$.
Let $\ell \geq 1$ be an integer such that
\begin{equation}
\label{eq:level}
	M2^{\ell - 1} \leq d < M2^{\ell}\quad.
\end{equation}
Then with probability $1 - \exp\pr{-C}$
we have
\begin{equation}
\label{eq:bound}
	\ell = 1\quad\text{or}\quad S_{\ell - 2} > M,\quad\text{and}\quad
	S_{\ell + 1} \leq M,
\end{equation}
and
\begin{equation}
\label{eq:approx}
	\abs{T_k - \mu_k (d - 1)} \leq \epsilon \mu_k (d - 1),
\end{equation}
where $k = \ell -1, \ell, \ell + 1$ and
	$\mu_k = 2^{-k}$.
\end{lemma}

\begin{proof}
First, note that Eq.~\ref{eq:level} implies
\begin{equation}
\label{eq:mub}
	 2\mu_{\ell + 1}d = \mu_\ell d < M \leq 4d\mu_{\ell + 1} = 2^{-1}d \mu_{\ell - 2}  \quad.
\end{equation}
To prove the lemma, 
let us define the events 
\[
	A_k =
	\abs{T_k - \mu_k (d - 1)} > \epsilon \mu_k (d - 1),
\]
and 
\[
	B_1 = S_{\ell - 2} \leq  M\quad\text{and}\quad B_2 = S_{\ell + 1} > M\quad.
\]
We will prove the result with union bound by showing that
\[
\begin{split}
	& P(A_{\ell - 1} \text{ or } A_{\ell} \text{ or } A_{\ell + 1} \text{ or } B_1  \text{ or } B_2)\\
	&\quad \leq P(A_{\ell - 1}) +  P(A_{\ell}) + P(A_{\ell + 1}) +P(B_1) + P(B_2) \\
	&\quad \leq 2/8 e^{-C} + 2/8 e^{-C} + 2/8 e^{-C} + 1/8 e^{-C} + 1/8 e^{-C} \quad.
\end{split}
\]

To bound $P(A_k)$, observe that $P(R_j \geq k) = \mu_k$. The
Chernoff bound now states that for $k \leq \ell + 1$ we have
\begin{align*}
	P(A_k) & = P\fpr{\abs{T_k - \mu_k (d - 1)} > \epsilon \mu_k (d - 1)} \\
	& < 2 \exp\pr{-\frac{\epsilon^2}{2 + \epsilon} \mu_k (d - 1)} & \text{(Eq.~\ref{eq:ch3})} \\
	& < 2 \exp\pr{-\frac{\epsilon^2}{2 + \epsilon} \mu_{\ell + 1} (d - 1)} & \text{($k \leq \ell + 1$)} \\
	& \leq 2\exp\pr{-\frac{\epsilon^2}{4(2 + \epsilon)} (M - 1)}  & \text{(Eq.~\ref{eq:mub}, $\mu_{\ell + 1} \leq 1/4$)} \\
	& = 2/8\exp\pr{-C}\quad.  & \text{(Eq.~\ref{eq:buffer})} \\
\end{align*}

Next, we bound $B_1$, assuming $\ell > 1$ as otherwise we can ignore the term, with
\begin{align*}
	P(S_{\ell - 2} \leq M) 
	& \leq P(S_{\ell - 2} \leq 2^{-1}\mu_{\ell - 2} d )  & \text{(Eq.~\ref{eq:mub})} \\
	& \leq P(S_{\ell - 2} \leq (1 - \epsilon)\mu_{\ell - 2} d ) & \text{($\epsilon \leq 1/2$)} \\
	& < \exp\pr{-\frac{\epsilon^2}{2} \mu_{\ell - 2}d}  & \text{(Eq.~\ref{eq:ch2})} \\
	& \leq \exp\pr{-\epsilon^2 M}  & \text{(Eq.~\ref{eq:mub})} \\
	& < \exp\pr{-\frac{\epsilon^2}{4(2 + \epsilon)} (M - 1)}  \\
	& = 1/8\exp\pr{-C}  & \text{(Eq.~\ref{eq:buffer})} \\
\end{align*}
and $B_2$ with
\begin{align*}
	P(S_{\ell + 1} > M) 
	& \leq P(S_{\ell + 1} > 2\mu_{\ell + 1} d )  & \text{(Eq.~\ref{eq:mub})} \\
	& \leq P(S_{\ell + 1} > (1 + 2\epsilon)\mu_{\ell + 1} d ) & \text{($\epsilon \leq 1/2$)} \\
	& < \exp\pr{-4\frac{\epsilon^2}{2 + 2 \epsilon} \mu_{\ell + 1}d}  & \text{(Eq.~\ref{eq:ch1})} \\
	& \leq \exp\pr{-\frac{\epsilon^2}{2 + 2 \epsilon} M}  & \text{(Eq.~\ref{eq:mub})} \\
	& < \exp\pr{-\frac{\epsilon^2}{4(2 + \epsilon)} (M - 1)}  \\
	& = 1/8\exp\pr{-C}\quad.  & \text{(Eq.~\ref{eq:buffer})} \\
\end{align*}

The bounds for $P(B_1)$, $P(B_2)$, and $P(A_k)$ complete the proof.
\end{proof}

\begin{proof}[Proof of Proposition~\ref{prop:approx}]
Let $S_i$, $T_i$, and $k$ be as defined in Definition~\ref{def:est} for $\est{\mathcal{R}; M}$.
Let $\ell$ be as defined in Eq~\ref{eq:level}.
We can safely assume that $M \leq d$.

Assume that the events in Lemma~\ref{lem:nice} hold.
Then Eq.~\ref{eq:bound} guarantees that $k = \ell - 1, \ell, \ell + 1$.

Write $Y_i = T_i2^i$ and $Z_i = M2^{i - 1}$.
Eq.~\ref{eq:approx} guarantees that
\begin{equation}
\label{eq:yibound}
	\abs{Y_i - (d - 1)} \leq \epsilon (d - 1)
\end{equation}
for $i = \ell - 1, \ell, \ell + 1$.
If $k = 0$ or $Y_k \geq Z_k$, then $\est{\mathcal{R}; M} = Y_k$ and we are done.

Assume $k > 0$ and $Y_k < Z_k$.
Then immediately
\[
	Z_k > Y_k \geq (1 - \epsilon) (d - 1) \quad.
\]

To prove the other direction, first assume that
$k > \ell - 1$. By definition of $k$, we have $S_{k - 1} > M$ and consequently $T_{k - 1} \geq M$. Thus, 
\[
	Z_k = M2^{k - 1} \leq T_{k - 1}2^{k - 1} = Y_{k - 1} \leq (1 + \epsilon)(d - 1),
\]
where the last inequality is given by Eq.~\ref{eq:yibound}.
On the other hand, if $k = \ell - 1$, then
\[
	Z_k = M2^{k - 1} = M2^{\ell - 2} \leq d/2 \leq d - 1 \leq (1 + \epsilon)(d - 1), 
\]
where the second inequality holds since $k > 0$ implies that $d \geq 2$.
In summary, Eq.~\ref{eq:appr} holds.

Since the events in Lemma~\ref{lem:nice} hold with probability of $1 - \exp\pr{-C}$, the claim follows.
\end{proof}

\bibliography{bibliography}

\end{document}